\pdfoutput=1
\documentclass[acmsmall,screen]{acmart}

\AtBeginDocument{%
    \providecommand\BibTeX{{%
        \normalfont B\kern-0.5em{\scshape i\kern-0.25em b}\kern-0.8em\TeX}}}

\setcopyright{acmlicensed}
\acmJournal{ToCT}
\acmYear{2024} \acmVolume{1} \acmNumber{1} \acmArticle{1} \acmMonth{1}\acmDOI{10.1145/3675415}

\acmSubmissionID{TOCT-D-22-00023}

\DeclareMathOperator{\Res}{Res}
\DeclareMathOperator{\PC}{PC}

\DeclareMathOperator{\Width}{Width}
\DeclareMathOperator{\Space}{Space}

\DeclareMathOperator{\Deg}{Deg}

\newcommand{\bbF}{{\mathbb{F}}}
\newcommand{\calC}{{\mathcal{C}}}
\newcommand{\calF}{{\mathcal{F}}}
\newcommand{\calG}{{\mathcal{G}}}
\newcommand{\calH}{{\mathcal{H}}}
\newcommand{\calN}{{\mathcal{N}}}

\newcommand{\ResXor}{\Res(\oplus)}

\newcommand{\PHP}{\mathrm{PHP}}
\newcommand{\HOLES}{\mathrm{HOLES}}
\newcommand{\Ordering}{\mathrm{Ordering}}
\newcommand{\ORDER}{\mathrm{ORDER}}
\newcommand{\DLO}{\mathrm{DLO}}
\newcommand{\WORDER}{\mathrm{WORDER}}
\newcommand{\GOP}{\mathrm{GOP}}
\newcommand{\FPHP}{\mathrm{FPHP}}
\newcommand{\NM}{\mathit{NM}}

\newcommand{\card}[1]{\left|#1\right|}

\begin{document}

\title{Resolution Over Linear Equations: Combinatorial Games for Tree-like Size and Space}

\author{Svyatoslav Gryaznov}
\authornote{The main part of this work was completed while the authors S.~Gryaznov and A.~Riazanov were affiliated to the St. Petersburg Department of Steklov Mathematical Institute of Russian Academy of Sciences (Russia).}
\email{svyatoslav.i.gryaznov@gmail.com}
\orcid{0000-0002-5648-8194}
\affiliation{%
    \institution{Imperial College London}
    \city{London}
    \country{United Kingdom}
}

\author{Sergei Ovcharov}
\email{sergei.d.ovcharov@gmail.com}
\orcid{0000-0002-9478-1949}
\affiliation{%
    \institution{St. Petersburg State University}
    \city{St. Petersburg}
    \country{Russia}
}

\author{Artur Riazanov}
\authornotemark[1]
\email{artur.riazanov@epfl.ch}
\orcid{0000-0001-7892-1502}
\affiliation{%
    \institution{\'{E}cole Polytechnique F\'{e}d\'{e}rale de Lausanne}
    \city{Lausanne}
    \country{Switzerland}
}


\begin{abstract}
    We consider the proof system $\ResXor$ introduced by Itsykson and Sokolov (Ann. Pure Appl. Log.'20), which is an extension of the resolution proof system and operates with disjunctions of linear equations over $\bbF_2$.

    We study characterizations of tree-like size and space of $\ResXor$ refutations using combinatorial games. Namely, we introduce a class of extensible formulas and prove tree-like size lower bounds on it using Prover--Delayer games, as well as space lower bounds. This class is of particular interest since it contains many classical combinatorial principles, including the pigeonhole, ordering, and dense linear ordering principles.
    Furthermore, we present the width-space relation for $\ResXor$ generalizing the results by Atserias and Dalmau (J. Comput. Syst. Sci.'08) and their variant of Spoiler--Duplicator games.
\end{abstract}




\maketitle

\section{Introduction}

The resolution proof system Res is among the most studied and well-known proof systems. Many state-of-the-art SAT solvers have resolution as their underlying proof system. An important subsystem of Res is tree-like resolution, which uses every derived clause at most once. Such proofs can be arranged into binary trees; hence the name. This proof system is primarily motivated by DPLL SAT solvers~\cite{Davis1960,Davis1962}, which construct tree-like refutations. However, we stress that modern SAT solvers employ other techniques and the resulting proofs are not tree-like.

The primary complexity measure for proof systems is the \emph{proof size}. In the context of SAT solvers it corresponds to the time required for the algorithm to verify a refutation.

Other important complexity measures are the \emph{proof space} and \emph{width}. The former, denoted by $\Space$, describes the memory required by a SAT solver. Formally, we will consider the \emph{clause-space}, that is the minimum number of clauses that must be retained in memory at any given time, and simply refer to it as ``space''. The latter, denoted by $\Width$, measures the number of variables that can appear in the clauses of a proof. Atserias and Dalmau~\cite{Atserias2008} proved the relation between these two measures for resolution and showed that for any $r$-CNF $\phi$ it holds that $\Space(\phi) \ge \Width(\phi) - r + 1$.

Lower bounds for many propositional proof systems are often formulated in terms of combinatorial games, and resolution is not an exception. Tree-like size lower bounds can be proved using Prover--Delayer games introduced by Impagliazzo and Pudl\'{a}k~\cite{Pudlak2000}. An unsatisfiable CNF is given to two players called Prover and Delayer. Delayer pretends to have a satisfying assignment and Prover tries to find a contradiction by asking for the values of the formula variables in an arbitrary order. Delayer can either assign the value or allow Prover to choose the value arbitrarily. In the latter case, Delayer earns a coin, and the number of earned coins corresponds to the logarithm of tree-like size.
Another game characterization we are interested in describes the space-width relation and was suggested by Atserias and Dalmau~\cite{Atserias2008}. They consider the existential $k$-pebble games, which are a particular type of Ehrenfeucht--Fra\"{i}ss\'{e} games. It is played by two players called Spoiler and Duplicator. The game was originally stated in the first order logic notation. We describe its adaptation to CNF formulas. Similarly to the previous game, Duplicator presumably has a satisfying assignment and Spoiler wants to obtain a contradiction by asking Duplicator about the values of the formula variables. Moreover, at most $k$ values can be retained in memory at any given time: Spoiler is forced to forget some values at some point. If Duplicator can win the game, i.e. Spoiler cannot arrive to a contradiction, then the proof width is at least $k+1$.

\paragraph{Resolution over linear extensions.}
Raz and Tzameret~\cite{DBLP:journals/apal/RazT08} introduced an extension of resolution with linear equations over the integers $\Res(\mathrm{lin}_{\mathbb{Z}})$. This extension empowers the resolution with the ability to count and thus allows polynomial-time refutations for hard tautologies based on counting like the pigeonhole principle ($\PHP$). Itsykson and Sokolov~\cite{Itsykson2014} considered a variant of this proof system $\ResXor$---or $\Res(\mathrm{lin}_{\bbF_2})$---that performs the counting modulo two. They proved several lower bounds on the tree-like variant of this system and on some of its subsystems. Later, Part and Tzameret~\cite{PartTzameret2018} studied the more general case of $\Res(\mathrm{lin}_{R})$, where $R$ is an arbitrary ring. In this paper we focus solely on $\ResXor$ and do not consider other rings.

\paragraph{$m$-extensible formulas and ordering principles.}
We start with generalizing the ideas behind the proof on the size of tree-like $\ResXor$ refutations for the weak pigeonhole principle $\PHP_n^m$~\cite{Itsykson2014} and proving size lower bounds for a vast class of so-called \emph{$m$-extensible formulas} w.r.t. a subset of formula clauses. Informally, it consists of formulas $\phi$ such that every linear system with a small number of equations that does not contradict any ``narrow'' clauses does not imply $\neg C$, where $C$ is a clause of $\phi$.
The formal definition is presented in Section~\ref{sec:framework}. We will sometimes call such formulas simply \emph{extensible} when we do not use the actual parameters. The proof is based on Prover--Delayer games, which were adapted for $\ResXor$ in~\cite{Itsykson2014}.

The class of extensible formulas includes the pigeonhole, ordering, and dense linear ordering principles. We show that the ordering principle gives a natural separation of resolution and tree-like $\ResXor$. Although such formulas were known before, the previous constructions require a non-trivial gadget and the proof followed from the communication complexity rather than being specific to $\ResXor$~\cite{Goos2014}. We also note that a slightly weaker lower bound for the ordering principle follows from the degree lower bound in Polynomial Calculus for the graph variant of the principle~\cite{Galesi2010,Nordstrom2015}. Our results for the dense linear ordering principle are new and are not known for stronger proof systems. In particular, its complexity for Polynomial Calculus remains an open problem. 

\paragraph{Space complexity of tree-like $\ResXor$ proofs.}
We also consider the space complexity of $\ResXor$-proofs and provide two different characterizations of it. It turns out that extensible formulas not only require large tree-like size, but also large space. We prove it by constructing a game characterization similar to Prover--Delayer games. In Section~\ref{sec:space_width} we consider another characterization of space using width, and extend Spoiler--Duplicator games to $\ResXor$.


\subsection{Our contribution}

\begin{enumerate}
    \item We introduce the class of extensible formulas and prove tree-like size and space lower bounds on formulas of this class.
    
    \item We prove that several well-known combinatorial principles belong to the classes of extensible formulas with appropriate parameters: the pigeonhole, ordering, and dense linear ordering principles. 
    
    \item We present the width-space relation for $\ResXor$ generalizing the results by Atserias and Dalmau~\cite{Atserias2008}.
\end{enumerate}

Some of these results appeared in Proceedings of CSR 2019~\cite{DBLP:conf/csr/Gryaznov19}.

\section{Preliminaries}\label{sec:prelim}

\subsection{Resolution over linear equations}

In this section we describe the proof system $\ResXor$ that we use throughout the paper~\cite{Itsykson2014}.

A \emph{linear clause} is a disjunction $\bigvee_{i=1}^m (f_i = \alpha_i)$ of linear equations ${\{f_i=\alpha_i\}}_{i=1}^m$ (linear literals), where $f_i$ are linear forms over $\bbF_2$ and $a_i \in \{0,1\}$ for $i \in [m]$.
We say that a linear clause $D$ is \emph{semantically implied} by linear clauses $C_1, \ldots, C_s$ if every satisfying assignment of $\bigwedge_{i=1}^s C_i$ also satisfies $D$. We denote it as $C_1, \ldots, C_s \vDash D$. A linear CNF formula is a conjunction of linear clauses.

This proof system has two rules:
\begin{enumerate}
  \item The \textbf{Resolution rule} allows deriving the linear clause $A \lor B$ from linear clauses $A \lor (f = 0)$ and $B \lor (f = 1)$, where $f$ is a linear form over $\bbF_2$.
  \item The \textbf{(Semantic) Weakening rule} allows deriving a linear clause $D$ from a linear clause $C$ if $D$ is semantically implied by $C$.
\end{enumerate}

A \emph{refutation} (proof of unsatisfiability) of a CNF formula $\phi$ is a derivation of the empty clause from the clauses of $\phi$ using these rules. A proof is called \emph{tree-like} if it can be arranged as a binary tree, where the root is labeled with the empty clause, the leaves are labeled with clauses of $\phi$, and every internal node is labeled with a result of an application of one of the rules to the children of this node.

Part and Tzameret~\cite{PartTzameret2018} considered similar generalizations over other fields (and rings) besides $\bbF_2$. Although these generalizations are not proof systems in the classical sense since the proofs cannot be verified in polynomial time~\cite{PartTzameret2018,DBLP:conf/csr/Gryaznov19}. In what follows we focus solely on the proof system $\ResXor$.

\subsection{Polynomial Calculus}

Another proof system we consider in our paper is Polynomial Calculus over a field $\bbF$~\cite{DBLP:conf/stoc/CleggEI96,DBLP:journals/siamcomp/AlekhnovichBRW02}, which we denote by $\PC_\bbF$. This proof system operates with multivariate polynomials over $\bbF$. Given a set of polynomials $\calF$, it uses the following derivation rules:
\begin{enumerate}
    \item The \textbf{Axiom download rule} allows deriving any axiom from $\calF$ or any Boolean axiom $x^2-x$, wher $x$ is a variable.
    \item The \textbf{Linear combination rule} allows deriving the polynomial $\alpha p + \beta q$ from polynomials $p$ and $q$, where $\alpha, \beta \in \bbF$.
    \item The \textbf{Multiplication by variable rule} allows deriving $xp$ from a polynomial $p$, where $x$ is a variable.
\end{enumerate}

A $\PC_{\bbF}$ refutation of $\calF$ is a derivation of the constant $1$ polynomial using the above rules.
The degree of a $\PC_\bbF$ refutation of $\calF$, denoted $\Deg(\calF)$, is the maximum degree of polynomials appearing in the refutation.
Any clause $C$ with $m$ literals can be expressed as a degree-$d$ polynomial using the standard arithmetization technique. Given a CNF formula $\phi$, a $\PC_\bbF$ refutation of $\phi$ is a refutation of the arithmetization of $\phi$.

In this paper we are only interested in the case $\bbF = \bbF_2$.

\subsection{Prover--Delayer games}

Pudl\'{a}k and Impagliazzo introduced in~\cite{Pudlak2000} a method for proving lower bounds on the size of tree-like resolution refutations using Prover--Delayer games. We incorporate its straightforward extension (see~\cite{Itsykson2017} for a detailed description of the method).

The game is played by two players: Prover and Delayer. They both have an unsatisfiable linear CNF formula $\phi$. Delayer pretends to know a satisfying assignment of $\phi$ and Prover wants to find a falsified clause by asking Delayer about the values of various linear forms over $\bbF_2$. The game consists of two steps:
\begin{enumerate}
  \item Prover chooses a linear form $g$ over $\bbF_2$ and asks Delayer for the value of $g$ on their assignment.
  \item Delayer answers with either a constant ($0$ or $1$) or the symbol $*$, allowing Prover to choose an arbitrary value. In the latter case, Delayer earns a coin.
\end{enumerate}

The game ends when there is a clause $C$ of $\phi$ which is falsified by every solution of the Prover's system.

\begin{lemma}[\cite{Itsykson2014,Itsykson2017}]\label{lem:proverdelayer}
  Let $\phi$ be an unsatisfiable linear CNF formula.
  If Delayer has a strategy guaranteeing a win of at least $k$ coins against any Prover in the game on $\phi$, then the size of any tree-like $\ResXor$ refutation of $\phi$ is at least $2^{k}$.
\end{lemma}

\subsection{Space complexity}
The space complexity measure was introduced by Esteban and Tor\'{a}n~\cite{DBLP:journals/iandc/EstebanT01} and later extended by Alekhnovich~et~al.~\cite{DBLP:journals/siamcomp/AlekhnovichBRW02}.
We introduce its generalization for $\ResXor$ proofs following~\cite{Itsykson2017}.
Let $\phi$ be an unsatisfiable linear CNF formula. A \emph{configuration} is a set of linear clauses. A refutation $\pi$ of $\phi$ is a sequence of configurations $S_1, \ldots, S_m$ such that $S_1$ is empty, $S_m$ contains the empty clause, and for every $i \in [m]$ the state $S_i$ is obtained from $S_{i-1}$ by the application of one of the following basic operations:
\begin{itemize}
    \item \textbf{Download operation}: $S_{i} \coloneqq S_{i-1} \cup \{C\}$, where $C$ is a clause of $\phi$.
    \item \textbf{Erasure operation}: $S_{i} \coloneqq S_{i-1} \setminus \{C\}$, where $C \in S_{i-1}$.
    \item \textbf{Inference operation}: $S_{i} \coloneqq S_{i-1}\cup\{C\}$, if $C$ can be deduced from the clauses in $S_{i-1}$ using the Resolution and Weakening rules.
\end{itemize}

The \emph{space} of a refutation $\pi = (S_1, \ldots, S_m)$ is defined as
\[
    \Space(\pi) \coloneqq \max_{i \in [m]} \card{S_i}.
\]

The \emph{space} of an unsatisfiable CNF formula $\phi$ is
\[
    \Space(\phi) \coloneqq \min_{\pi} \Space(\pi),
\]
where the minimum is taken over all $\ResXor$ refutations of $\phi$.

\section{\texorpdfstring{$m$}{m}-extensible formulas}\label{sec:framework}

We use a crucial fact about $\bbF_2$ which distinguishes it from other rings and makes the analysis of $\ResXor$ proofs easier. That is, a disequality $f(x) \neq \alpha$ is equivalent to the equality $f(x) = 1 - \alpha$. Due to this fact, it is more convenient to work with \emph{negations} of linear clauses: the negation $\neg C$ of a linear clause $C = \bigvee_{i=1}^m (f_i = \alpha_i)$ is the linear system $\bigwedge_{i=1}^m (f_i = 1 - \alpha_i)$. The correspondence between a linear clause and its negation is the following: a variable assignment is a solution of $\neg C$ if and only if it falsifies $C$.

Itsykson and Sokolov proved a tree-like size lower bound for the weak pigeonhole principle $\PHP_n^m$, which was later generalized by Part and Tzameret~\cite{PartTzameret2018} for arbitrary rings. These proofs rely on particular properties of the solution space of $\PHP$: if an assignment $\alpha$ satisfies all the ``short'' clauses of $\PHP$ (in this case, the hole clauses), then so does every $\beta \le \alpha$, i.e. we can safely replace any $1$ with $0$ in $\alpha$.

Consider an unsatisfiable linear CNF formula $\phi$ in the variables $v$. Let $F$ be a subset of the clauses of $\phi$. A solution $\sigma$ of a linear system $Av=b$ is \emph{$F$-proper} if for every clause $C \in F$ the solution $\sigma$ satisfies $C$. We say that $\phi$ is \emph{$m$-extensible w.r.t. $F$} if for every linear system $Av=b$ over $\bbF_2$ with less than $m$ equations, which has an $F$-proper solution, and every $C \in \phi \setminus F$, there exists an $F$-proper solution of $Av=b$ that satisfies $C$.

The following theorem is a direct generalization of Lemma~3.5 from~\cite{Itsykson2014}. We present its proof here for the sake of completeness.

\begin{theorem}[cf.~\cite{Itsykson2014}]\label{thm:main}
    Let $\phi$ be a linear CNF formula. If there exist a number $m$ and a subset of clauses $F \subseteq \phi$ such that $\phi$ is $m$-extensible w.r.t. $F$, then the size of any tree-like $\ResXor$ refutation of $\phi$ is at least $2^{m}$.
\end{theorem}
\begin{proof}
    By Lemma~\ref{lem:proverdelayer}, it is enough to construct a strategy for Delayer to earn at least $m$ coins.

    Let $\Phi$ be the system of all linear forms corresponding to Prover's questions together with its assigned values, i.e. Delayer's answers or Prover's choices when Delayer answered with $*$. We gradually construct $\Phi$ while keeping the following invariant: if Delayer earned less than $m$ coins, then the system $\Phi$ has an $F$-proper solution.

    Initially, the system is empty.
    Suppose that a new round of the game starts with Prover asking about the value of a linear form $g$. If $g(x)=\alpha$ is implied by all $F$-proper solutions of $\Phi$ for some $\alpha \in \bbF_2$, then Delayer simply answers $\alpha$. Otherwise, $g$ is not constant on the set of all $F$-proper solutions of $\Phi$. In that case, Delayer answers $*$ and allows Prover to choose the value arbitrarily.

    Consider three possible conclusions of the game:
    \begin{enumerate}
        \item $\Phi$ does not have any solutions. By invariant, this case is only possible if Delayer has already earned at least $m$ coins.
        \item $\Phi$ only has solutions that falsify clauses from $F$. Similarly to the previous case, Delayer has earned at least $m$ coins.
        \item There is a clause $C \in \phi \setminus F$ that is falsified by all solutions of $\Phi$. In particular, every $F$-proper solution of $\Phi$ falsifies $C$. In this case, we consider the linear system $\Phi' \subseteq \Phi$ corresponding to the answers $*$ of Delayer. By construction, the lines of $\Phi \setminus \Phi'$ must be implied by all $F$-proper solutions of $\Phi'$. Hence, $\Phi'$ does not have any new $F$-proper solutions and $C$ is also falsified by every $F$-proper solution of $\Phi'$.
        Assume that $\Phi'$ has less than $m$ lines. Then, by invariant, it has an $F$-proper solution. Since $\phi$ is $m$-extensible, $\Phi'$ has an $F$-proper solution that satisfies $C$, contradicting the assumption.
        Thus, $\Phi'$ has at least $m$ lines and Delayer has earned at least $m$ coins, each corresponding to a line of $\Phi'$.
    \end{enumerate}
\end{proof}

The following result provides a method for proving lower bounds on the space complexity of extensible CNF formulas. The assumptions of the following theorem are exactly the same as for Theorem~\ref{thm:main}. The proof is not stated in terms of games and is rather straightforward.

\begin{theorem}\label{thm:space}
    Let $\phi$ be a linear CNF formula. If there exist a number $m$ and a subset of clauses $F \subseteq \phi$ such that $\phi$ is $m$-extensible w.r.t. $F$, then the space required by any $\ResXor$ refutation of $\phi$ is at least $m$.
\end{theorem}
\begin{proof}
    Assume for the sake of contradiction that there is a proof $S_1, \ldots, S_t$ of $\phi$ requiring less than $m$ clauses in memory. We construct a set of assignments $\sigma_1, \ldots, \sigma_t$ such that $\sigma_i$ satisfies all clauses from $F$ and $S_i$. Since $S_t$ consists of the empty clause, this gives us a contradiction.
    For the base case, choose any assignment that does not falsify clauses from $F$. It exists by $m$-extensibility of $\phi$ w.r.t. $F$ applied to the empty system.
    For the inductive step, there are several possibilities:
    \begin{enumerate}
        \item
        $S_i = S_{i-1} \setminus \{C\}$ by an erasure operation. In this case simply set $\sigma_{i} \coloneqq \sigma_{i-1}$.

        \item
        $S_i = S_{i-1} \cup \{C\}$ by an inference operation. Since the inference rules are sound, $C$ is satisfied by $\sigma_{i-1}$ and we can set $\sigma_{i} \coloneqq \sigma_{i-1}$.

        \item
        $S_i = S_{i-1} \cup \{C\}$ by a download operation. If $C \in F$, then we can set $\sigma_{i} \coloneqq \sigma_{i-1}$ since $\sigma_{i-1}$ satisfies all clauses in $F$. If $C \in \phi \setminus F$, we construct a new assignment as follows. Let $S_{i-1} = \{C_1, \ldots, C_k\}$,  where $k < m$. Since $\sigma_{i-1}$ satisfies every clause in $S_{i-1}$, for every $j \in [k]$, there exists a linear literal $(\ell_j = \alpha_j) \in C_j$ that is satisfied by $\sigma_{i-1}$. Consider a linear system $\bigwedge_{j=1}^k (\ell_j=\alpha_j)$. This system has an $F$-proper solution $\sigma_{i-1}$ and contains $k < m$ lines. Thus, by $m$-extensibility, it has an $F$-proper solution $\tau$ that satisfies $C$. By construction, $\tau$ satisfies $S_i$ and we can set $\sigma_i \coloneqq \tau$.
    \end{enumerate}
\end{proof}

\subsection{Applications}\label{sec:applications}

In this section we prove lower bounds for the pigeonhole, ordering, and dense linear ordering (DLO) principles using the above method.

\subsubsection{Pigeonhole principle}

The formula $\PHP_n^m$ encodes the (weak) pigeonhole principle. $\PHP_n^m$ includes the variables $p_{ij}$, where $i \in [m]$ and $j \in [n]$; $p_{ij}$ states that the $i$th pigeon flies to the $j$th hole. It has two types of clauses:
\begin{enumerate}
  \item Pigeon axioms: $P_{i} \coloneqq \bigvee_{j \in [n]} p_{ij}$, for all $i \in [m]$.
  \item Hole axioms: $H_{ijk} \coloneqq \neg p_{ik} \lor \neg p_{jk}$, for all $i \neq j \in [m]$ and all $k \in [n]$.
\end{enumerate}

Define $\HOLES_n^m$ as the set of the hole axioms $\{H_{ijk} : i \neq j \in [m], k \in [n]\}$. In this case $\HOLES_n^m$-proper solutions encode partial injective mappings of pigeons to holes.

A weaker variant of the following lemma was proved in~\cite{Itsykson2014}. It was later improved by Garl\'{\i}k. For the sake of completeness, we present the proof here.

\begin{lemma}[M.~Garl\'{\i}k, personal communication]\label{lem:small_system_php}
    $\PHP_n^m$ is $(n-1)$-extensible w.r.t. $\HOLES_n^m$.
\end{lemma}
\begin{proof}
    We need the following important property of $\PHP$: if in a $\HOLES_n^m$-proper solution we change any $1$ to $0$, then the resulting assignment remains $\HOLES_n^m$-proper.

    Consider a linear system $Ap=b$ in the variables of $\PHP_n^m$ with $k \le n-1$ lines.
    We start by proving that there is a $\HOLES_n^m$-proper solution that has at most $k$ ones. Let $\sigma$ be a $\HOLES_n^m$-proper solution with the minimum number of ones and define $I \coloneqq \{(s,t) : \sigma(p_{st}) = 1\}$. Assume, for the sake of contradiction, that $\sigma$ has strictly more than $k$ ones. Consider the linear system $A_I p = 0$, where $A_I$ is the matrix consisting of the columns of $A$ with indices from $I$ (all variables outside $I$ are assigned to zeroes). Since $\card{I} > k$, this system has a non-trivial solution $\tau$. By the property above, the solution $\sigma + \tau$ is a $\HOLES_n^m$-proper solution of $Ap=b$. Since $\tau$ is non-trivial, $\sigma+\tau$ has fewer ones than $\sigma$ contradicting the mimimality assumption.

    Consider a pigeon axiom $P_i$.
    If $i \in I$ then $\sigma$ is a $\HOLES_n^m$-proper solution satisfying $P_i$ and we are done.
    Otherwise, since $\sigma$ has at most $k$ ones, it places no more than $k$ pigeons and there exist empty holes $j_1, \ldots, j_{n-k}$. Let $J \coloneqq \{(i,j_1), \ldots, (i, j_{n-k})\}$ and consider the homogenous system $A_{I \cup J} p = 0$. This system is underdetermined, thus it has a non-trivial solution $\rho$. Then $\sigma+\rho$ is a $\HOLES_n^m$-proper solution of $Ap=b$. Furthermore, there exists $(i,j) \in J$ such that $\rho(p_{ij}) = 1$ since otherwise it contradicts the minimality assumption. Hence, $\sigma+\rho$ places the $i$th pigeon to the $j$th empty hole and satisfies $P_i$.
\end{proof}

\begin{theorem}
  The size of any tree-like $\ResXor$ refutation of $\PHP_n^m$ is at least $2^{n-1}$.
  
  The space of any $\ResXor$ refutation of $\PHP_n^m$ is at least $n-1$. Furthermore, the space lower bound is tight.
\end{theorem}
\begin{proof}
  Follows from application of Lemma~\ref{lem:small_system_php} to Theorems~\ref{thm:main} and~\ref{thm:space}.
\end{proof}

\subsubsection{Ordering principle}\label{ssec:ord}

The ordering principle says that every finite linearly ordered set has a minimum.
We encode it by an unsatisfiable CNF formula $\Ordering_n$ in the variables ${(x_{ij})}_{i \neq j \in [n]}$, where $x_{ij}$ states that the $i$th element is less than the $j$th element (for clarity, we denote it as $i \prec j$). The clauses of $\Ordering_n$ are the following:
\begin{enumerate}
  \item Anti-symmetry: $\neg x_{ij} \lor \neg x_{ji}$, for all $i \neq j \in [n]$.
  \item Totality: $x_{ij} \lor x_{ji}$, for all $i \neq j \in [n]$.
  \item Transitivity: $\neg x_{ij} \lor \neg x_{jk} \lor x_{ik}$, for all distinct $i,j,k \in [n]$.
  \item Non-minimality (i.e. non-existence of the minimum element):
    \[
      \NM_i \coloneqq \bigvee_{j \in [n]\setminus \{i\}} x_{ji}, \text{ for all } i \in [n].
    \]
\end{enumerate}

Define $\ORDER_n$ as the set of clauses that encode the properties of a linear order (types 1--3). The rest of the clauses are the non-minimality clauses $\{\NM_i : i \in [n]\}$.

\begin{lemma}\label{lem:small_system_ord}
    $\Ordering_n$ is $(n-2)$-extensible w.r.t. $\ORDER_n$.
\end{lemma}
\begin{proof}
    We prove the statement by induction on $n$.

    Let $Ax=b$ be a linear system with less than $n-2$ lines with an $\ORDER_n$-proper soluton $\sigma$.
    $\ORDER_n$-proper solutions of $Ax=b$ encode linear orders. Since apart from the minimum, every element of a linear order has a predecessor, any $\ORDER_n$-proper solution falsifies exactly one non-minimality clause that encodes the non-minimality of the minimum. 

    Without loss of generality, we can assume that $\sigma$ is the ordering of the set $[n]$ and encodes the order $1 \prec 2 \prec \cdots \prec n$ (we can simply rename the indices).
    We also replace all occurrences of $x_{ji}$, where $j \succ i$, with $1-x_{ij}$. Since we are only interested in $\ORDER_n$-proper solutions, this operation is safe.

    The base case $n=2$ is trivial. The system $Ax=b$ is empty, thus, the encoding of every linear order on the set of $n$ elements is an $\ORDER_n$-proper solution of the system. Clearly, for $n>1$ there are at least two of them with different minimum elements.

    For the inductive step
    we choose $i$ and $j$, where $i \prec j$, satisfying the following properties:
    \begin{enumerate}
        \item $x_{ij}$ appears in some equation in $Ax=b$ with a non-zero coefficient.
        \item For all $k, l$ such that $i \preceq k \prec l \preceq j$ and $(k,l) \neq (i,j)$, $x_{kl}$ does not appear in any line of $Ax=b$ with a non-zero coefficient.
        \item $i$ is the largest possible: if $i'$ and $j'$ satisfy conditions 1--2, then either $(i',j') = (i,j)$ or $i \succ i'$.
    \end{enumerate}

    Note that such $i$ and $j$ exist since any $x_{ij}$ in the system with the smallest value of $j-i$ satisfies conditions 1--2.

    Suppose that $i=1$. By construction (property 3), only variables $y = (x_{12}, \ldots, x_{1n})$ appear in the system $Ax=b$.
    Hence, it can be written as $A'y = b$, where $A'$ consists of the columns of $A$ that correspond to $y$. The new system has $n-1$ variables. Let $\sigma'$ be the projection of $\sigma$ on $y$. Since the element $1$ is the minimum in $\sigma$, $\sigma' = (1, \ldots, 1)$.
    The number of equations in $A'y=b$ is at most $n-2$ that is fewer than the number of variables. Hence, the system must have another solution $\tau' \neq (1, \ldots, 1)$. Since these variables only encode the position of the element $1$ in the order, this partial solution can be easily extended to an $\ORDER_n$-proper solution of $Ax=b$, where the position of the element $1$ is fixed according to $\tau'$. Let $\tau$ be such an extension. Since there are at least one zero in $\tau'$, the element $1$ in $\tau$ cannot be the minimum.

    Now consider the case $i>1$. In this case we ``glue'' the elements $i$ and $j$ together and then reduce the problem to the induction hypothesis.

    We start by changing the system such that exactly one of its equations contains the variable $x_{ij}$. Consider any equation of $Ax=b$ with $x_{ij}$ and add it to the other ones that also contain $x_{ij}$. We denote this equation with $x_{ij}$ as $f(x)=\alpha$ and let $Bx=c$ be the rest of the system.

    We ``glue'' elements $i$ and $j$ together in $Bx=c$ by syntactically replacing all the occurrences of $j$ by $i$. This operation is correct since there are no $x_{kl}$ for all $i \preceq k \prec l \preceq j$ in $Bx=c$ by the choice of $i,j$ and the fact that $Bx=c$ does not contain $x_{ij}$. Let $B'x' = c$ be the system with $i$ and $j$ ``glued'' together. It has an $\ORDER_{n-1}$-proper solution $\sigma'$ that encodes the ordering $1 \prec \cdots \prec j-1 \prec j+1 \prec \cdots \prec n$ of the set $[n] \setminus \{j\}$.

    The system $B'x'=c$ has at most $n-3$ equations, depends on the variables corresponding to the order on the $(n-1)$-element set $[n] \setminus \{j\}$, and has an $\ORDER_{n-1}$-proper solution $\sigma'$. Therefore, we can apply the induction hypothesis and get a different $\ORDER_{n-1}$-proper solution $\tau'$ corresponding to the order $a_1 \prec \cdots \prec a_{n-1}$ with $a_1 \neq 1$.

    Let $i=a_k$, $k \in [n-1]$. We consider two orders
    \begin{equation}
    a_1 \prec \cdots \prec a_{k-1} \prec i \prec j \prec a_{k+1} \prec \cdots, a_{n-1}
    \end{equation}
    and
    \begin{equation}
    a_1 \prec \cdots \prec a_{k-1} \prec j \prec i \prec a_{k+1} \prec \cdots, a_{n-1}.
    \end{equation}
    They extend $\tau'$ to $[n]$ by placing $j$ immediately after or before $i$. Let $\tau^{(1)}$ and $\tau^{(2)}$ be their encodings as $\ORDER_n$-proper solutions of $Bx=c$. Exactly one of them satisfies the equation $f(x) = \alpha$ since $x_{ij}$ appears in it. Therefore, there exists $t \in {\{1,2\}}$ such that $\tau^{(t)}$ is an $\ORDER_n$-proper solution of $Ax=b$, and $\sigma$ and $\tau^{(t)}$ have different minimum elements in their corresponding orders.
\end{proof}

\begin{theorem}\label{thm:ordering_bound}
    The size of any tree-like $\ResXor$ refutation of $\Ordering_n$ is at least $2^{n-2}$.
    The space of any $\ResXor$ refutation of $\Ordering_n$ is at least $n-2$.
\end{theorem}
\begin{proof}
    Follows from application of Lemma~\ref{lem:small_system_ord} to Theorems~\ref{thm:main} and~\ref{thm:space}.
\end{proof}

\subsubsection{Dense Linear Ordering principle}\label{ssec:dlo}

The dense linear ordering principle states that a finite linearly ordered set cannot be dense. It was originally suggested by Urquhart and used by Riis in \cite{DBLP:journals/cc/Riis01}. We encode it by an unsatisfiable CNF formula $\DLO_n$ in the variables ${(x_{ij})}_{i \neq j \in [n]}$, where $x_{ij}$ states that the $i$th element is less than the $j$th element ($i \prec j$), and ${(z_{ikj})}_{i \neq j \neq k \in [n]}$, where $z_{ikj}$ encodes that $k$ is a ``witness'' of $i \prec j$, i.e. if $i \prec j$ and $z_{ikj}=1$ then $i \prec k \prec j$. The clauses of $\DLO_n$ are the following:
\begin{enumerate}
    \item Anti-symmetry: $\neg x_{ij} \lor \neg x_{ji}$ for all $i \neq j \in [n]$.
    \item Totality: $x_{ij} \lor x_{ji}$ for all $i \neq j \in [n]$.
    \item Transitivity: $\neg x_{ij} \lor \neg x_{jk} \lor x_{ik}$ for all distinct $i,j,k \in [n]$.
    \item Semantics for $z$'s: $\neg z_{ikj} \lor x_{ik}$ and $\neg z_{ikj} \lor x_{kj}$ for all distinct $i,j,k \in [n]$.
    \item Density (if $i \prec j$ then there must be some $k$ between them):
        \[
            D_{ij} = \neg x_{ij} \lor (\bigvee_{k \in [n] \setminus {\{i,j\}}} z_{ikj}) \text{ for all } i \neq j \in [n].
        \]
\end{enumerate}

Similarly to the ordering principle, we define $\WORDER_n$ (``W'' here stands for ``witness'') as the set of clauses that encode the properties of a linear order (types 1--3) together with the clauses describing the semantics for $z$'s (type 4). The remaining clauses are the density clauses ${\{D_{ij}\}}_{i \neq j \in [n]}$.

\begin{lemma}\label{lem:small_system_dlo}
    $\DLO_n$ is $\left\lfloor \frac{n-3}{3} \right\rfloor$-extensible w.r.t. $\WORDER_n$.
\end{lemma}

We need the following restriction of this lemma.
\begin{proposition}\label{lem:dlo_helper}
    Let $Av=b$ be a linear system in the variables $v = (x_{ij}) \cup (z_{ikj})$ with at most $n-3$ equations that has a $\WORDER_n$-proper solution $\sigma$, which sets every $z_{ikj}$ to zero. Then for every density clause $D_{st}$ there is a $\WORDER_n$-proper solution $\tau$ of $Av=b$ that satisfies $D_{st}$ and $\tau(z_{ikj}) = 1$ only if $i=s$ and $j=t$.
\end{proposition}
\begin{proof}
    Without loss of generality, we can assume that $\sigma$ is the ordering of the set $[n]$ and encodes the order $1 \prec 2 \prec \cdots \prec n$ (we can simply rename the indices).
    We replace all occurrences of $x_{ji}$, where $j \succ i$ in the system with $1-x_{ij}$, as we did in the proof for the ordering principle.

    Let $D_{st}$ be a density clause of $\phi$ that is falsified by $\sigma$.

    We prove the statement by induction on $n$.

    The base case $n=3$ is trivial since the system is empty: we can either choose an order with $s \succ t$ or choose the order $s \prec w \prec t$, where $w \neq s$ and $w \neq t$, and set $z_{swt}=1$.

    For the inductive step
    we choose $i$ and $j$, where $i \prec j$, satisfying the following properties:
    \begin{enumerate}
        \item $x_{ij}$ appears in $Av=b$ with a non-zero coefficient.
        \item For all $k, l$ such that $i \preceq k \prec l \preceq j$ and $(k,l) \neq (i,j)$, $x_{kl}$ does not appear in any line of $Av=b$ with a non-zero coefficient.
    \end{enumerate}

    If there are $i$ and $j$ satisfying this conditions and $t \preceq i$, then we ``glue'' $i$ and $j$ into $i$ as we did in the proof of Lemma~\ref{lem:small_system_ord}: we can leave $x_{ij}$ in exactly one equation $f(v) = \alpha$ of $Av=b$ by adding some equations and consider the linear system $Bv=c$ excluding this equation. Then we 
    syntactically replace each occurrence of $j$ in $Bv=c$ with $i$ and denote the resulting system as $B'v'=c$. The system $B'v'=c$ has a $\WORDER_n$-proper solution $\tau$, which satisfies $D_{st}$, by the induction hypothesis. Now we can extend $\tau$ by adding $j$ and placing immediately before of after $i$. Exactly one of such extensions satisfies $f(v)=\alpha$ and we are done.

    Otherwise, $t$ is not being compared with any other $w$, where $w \succ t$, and we can modify $\sigma$ so that $t$ is the maximum element in the order encoded by $\sigma$. Consider an assignment $\sigma'$ such that $\sigma'(x_{wt}=1)$ for every $w \succ t$ and $\sigma'$ coincides with $\sigma$ on the other variables. Since for all $w \succ t$, $x_{wt}$ does not appear in the system, $\sigma'$ is a solution of $Av=b$ as well.

    Similarly, we try to choose a pair $(i,j)$ ``to the left'' of $s$. If there is $i$ and $j$ satisfying aforementioned properties and $j \preceq s$, then we can do the similar reasoning: we ``glue'' $i$ and $j$ into $j$ and apply the induction hypothesis. If we cannot choose such a pair, we can modify the solution such that $s$ is the minimum.

    For the remaining case, we showed that there is another solution $\sigma'$ that has $s$ as the minimum and $t$ as the maximum. Hence, any other element can be taken as a witness of $s \prec t$. Assign the value from $\sigma$ for every variable except for $z_{skt}$ for all $k$. The resulting system has $n-2$ variables, at most $n-3$ equations, and has an all-zero solution. Thus, it has at least one other solution with some $z_{skt}$ set to $1$. We can change $\sigma'$ according to this solution and get the resulting $\WORDER_n$-proper solution $\tau$, which has $\tau(z_{skt})=1$ for some $k$.
\end{proof}

Now we are ready to prove Lemma~\ref{lem:small_system_dlo}.

\begin{proof}[Proof of Lemma~\ref{lem:small_system_dlo}]
    Consider a linear system $Av=b$ with at most $M \coloneq \left\lfloor \frac{n-3}{3} \right\rfloor$ equations that has a $\WORDER_n$-proper solution.
    Let $\sigma$ be a $\WORDER_n$-proper solution of the system with the minimum number of ones assigned to variables ${\{z_{ikj}\}}_{i,k,j}$. We claim that $\sigma$ sets at most $M$ ones to these variables. Assume the opposite and fix all variables in $Av=b$ according to $\sigma$ except for $z_{ikj}$ satisfying $\sigma(z_{ikj})=1$. It has at most $M$ lines, hence there is a solution that has at most $M$ ones. Here we stress the fact that if $\sigma$ is an $F$-proper solution with $\sigma(z_{ijk})=1$ for some $i,j,k$, then we can set $\sigma(z_{ijk})=0$ and it will remain $F$-proper, i.e. we can safely change the value of any variable $z_{ikj}$ from $1$ to $0$.

    Let us define another system $Cv=d$.
    It contains all equations from $Av=b$ with $z_{ikj}$ replaced by $\sigma(z_{ikj})$. In addition, for every $z_{ikj}$ satisfying $\sigma(z_{ikk})=1$, it contains two additional equations $x_{ik}=1$ and $x_{kj}=1$.
    This system has as most $3M \le n-3$ equations and has an $F$-proper solution $\sigma'$ that coincides with $\sigma$ on $x_{ij}$ and assigns every $z_{ikj}$ to $0$.
    We can apply Proposition~\ref{lem:dlo_helper} and get an $F$-proper solution $\tau$ of $Cv=d$ that satisfies $D_{st}$. Since for every $z_{ikj}$ in $Av=b$ with $\sigma(z_{ikj})=1$ we have $\tau(x_{ik})=1$ and $\tau(x_{kj})=1$, we can set $\tau(z_{ikj})\coloneqq 1$ and it will remain $F$-proper. Thus, we can change $\tau$ such that it is an $F$-proper solution of $Av=b$ and satisfies $D_{st}$.
\end{proof}

\begin{theorem}
    The size of any tree-like $\ResXor$ refutation of $\DLO_n$ is at least $2^{\left\lfloor \frac{n-3}{3} \right\rfloor}$.
    The space of any $\ResXor$ refutation of $\DLO_n$ is at least $\left\lfloor \frac{n-3}{3} \right\rfloor$.
\end{theorem}
\begin{proof}
    Follows from application of Lemma~\ref{lem:small_system_dlo} to Theorems~\ref{thm:main} and~\ref{thm:space}..
\end{proof}

We emphasize here that this result is new and to the moment cannot be obtained using different techniques.

\section{Width and Space}\label{sec:space_width}

Spoiler--Duplicator games were used by Atserias and Dalmau in~\cite{Atserias2008} for proving resolution space lower bounds in terms of width. In this section we extend these games to the $\ResXor$ proof system and show that similar lower bounds hold for it.

We define the width of a linear clause similarly to the resolution width as the number of linear literals in the clause. The width of a refutation of an unsatisfiable CNF $\phi$ is the maximum width of clauses in the refutation. The width of $\phi$, denoted as $\Width(\phi)$, is the minimum width of its refutations.

The following definition is an extension of Definition~2 from \cite{Atserias2008} to the $\ResXor$ proof system. Let $\phi$ be a linear CNF formula and $\calH$ a non-empty family of linear systems over $\bbF_2$ in the variables of $\phi$. We say that $\calH$ is a \emph{$k$-winning strategy} if it satisfies the following properties:
\begin{enumerate}
    \item For every $\calF \in \calH$, $\calF$ contains at most $k$ equations (denoted as $\card{\calF} \leq k$).
    \item For every $\calF \in \calH$ and every linear clause $C$ in $\phi$, there exists a solution of $\calF$ that satisfies $C$.
    \item If $\calG$ is a linear system satisfying $\calF \vDash \calG$ for some $\calF \in \calH$ and $\card{\calG} \leq k$, then $\calG \in \calH$.
    \item For every $\calF \in \calH$ with $\card{\calF} < k$ and for every linear form $f$, there exists a constant $a \in \bbF_2$ such that $\calF \land (f = a) \in \calH$.
\end{enumerate}
\label{ws}

\begin{lemma}\label{lem:sw_l2}
    Let $\phi$ be a linear $r$-CNF formula. Suppose that there are no $\ResXor$ refutations of $\phi$ of width $k\geq r$. Then there is a $k$-winning strategy for $\phi$.
\end{lemma}
\begin{proof}
    Let $\calC = \{C_1, \ldots, C_m\}$ be the set of all linear clauses having a $\ResXor$ derivation from $\phi$ of width at most $k$. Let $\calH$ be the set of all linear systems $\calF$ such that $\card{\calF} \leq k$ and for every $C_i$, $i \in [m]$, there exists a solution of $\calF$ that satisfies $C_i$. We show that $\calH$ is a $k$-winning strategy.
    
    The set $\calC$ does not contain the empty clause since $\phi$ does not have refutations of width at most $k$. Hence, every clause in $\calC$ can be satisfied and $\calH$ contains the empty system.
    It is clear that $\calH$ satisfies properties 1--3 from the definition of the winning strategy. Consider the fourth property: we have a linear system $\calF \in \calH$ with $\card{\calF} < k$ and a linear form $f$. Suppose that there are no valid extensions of $\calF$ to $f$ in $\calH$. Hence, there exists a linear clause $C \in \calC$ falsified by every solution of $\calF_0 = \calF \land (f=0)$ and a linear clause $D \in \calC$ falsified by every solution of $\calF_1 = \calF \land (f=1)$. Equivalently, $C \vDash \lnot \calF \lor (f=1)$ and $D \vDash \lnot \calF \lor (f=0)$. It implies that both $\lnot \calF \lor (f=0)$ and $\lnot \calF \lor (f=1)$ have derivations of width $k$. Then $\lnot \calF$ also has a derivation of width $k$. Therefore, $\lnot \calF$ lies in $\calC$. However, it means that $\calF$ has a solution that satisfies $\lnot \calF$, which is a contradiction.
\end{proof}

\begin{lemma}\label{l3}
    If a linear CNF formula $\phi$ has a $(k+1)$-winning strategy $\calH$, then $\phi$ does not have a $\ResXor$ refutation of width $k$.
\end{lemma}
\begin{proof}
    Suppose such a refutation exists. We show by induction that for every system $\calF \in \calH$ and every linear clause $C$ from this refutation, there exists a solution of $\calF$ that satisfies $C$. Thus, the refutation cannot contain the empty clause, which gives a contradiction. The base case: if $C$ is the initial clause, then the statement is satisfied by definition of a winning strategy.
    
    For the inductive step we have to consider two cases.
    \begin{enumerate}
        \item
        $D$ is the result of an application of the weakening rule. Thus, $C \vDash D$ for some clause $C$ that was derived earlier. However, the solutions of $\calF$ that satisfy $C$ also satisfy $D$. Hence, the statement is also true for $D$.

        \item
        $C \lor D$ is the result of an application of the resolution rule to the clauses $C \lor (f = 0)$ and $D \lor (f = 1)$. Suppose that $\calF \vDash \lnot (C \lor D)$, i.e. $\calF$ does not have a solution that satisfies $C \lor D$. Note that $\lnot(C \lor D)$ is a linear system that contains at most $k$ equations. By the third property of a winning strategy, $\lnot (C \lor D) \in \calH$. Hence, $\lnot(C \lor D) \land (f = \alpha) \in \calH$ for some $\alpha \in \mathbb{F}_2$. By the induction hypothesis, this system has a solution $\sigma_0$ that satisfies $C \lor (f=0)$ and a solution $\sigma_1$ that satisfies $D \lor (f=1)$. However, if $\alpha = 0$, then every solution of $\lnot(C \lor D) \land (f = \alpha)$ falsifies $D \lor (f =1)$. Analogously, if $\alpha = 1$, then every solution of this system falsifies $C \lor (f=0)$, which gives a contradiction.
    \end{enumerate}
\end{proof}

We need the following well-known fact about linear algebra.
\begin{lemma}\label{fact:well_known}
    Let $Ax = b$ be a consistent linear system over a field $F$ and $f(x) = \alpha$ a linear equation over $F$. Then $f(x) = \alpha$  is semantically implied by $Ax = b$ if and only if $f(x) = \alpha$ is a linear combination of equations of $Ax = b$.
\end{lemma}

Itsykson and Sokolov~\cite{Itsykson2014} showed that the semantic weakening rule can be simulated by three syntactic rules which are more convenient to work with. The following lemma shows that such simulation does not change space complexity significantly.
\begin{lemma}\label{lem:prop2}
    The weakening rule can be simulated by the following syntactic rules: 
    \begin{enumerate}
        \item \emph{The simplification rule} that allows deriving $D$ from $D \lor (0 = 1)$.
        \item \emph{The syntactic weakening rule} that allows deriving $D \lor (f = \alpha)$ from $D$.
        \item \emph{The addition rule} that allows deriving $D\lor(f_1 = \alpha_1)\lor(f_1+f_2 = \alpha_1+\alpha_2+1)$ from $D\lor(f_1 = \alpha_1)\lor(f_2 = \alpha_2)$.
    \end{enumerate}

    The space of the new proof increases by at most $1$.
\end{lemma}
\begin{proof}
    This construction is motivated by Proposition~2.6 from~\cite{Itsykson2014}. Here we think about linear clauses as negotiations of linear systems. Let $\Phi$ be a linear system and consider the linear clause $\neg \Phi$. In this representation the syntactic weakening rule allows adding new equations to $\Phi$, the addition rule allows adding one equation in $\Phi$ to another, and the simplification rule allows removing trivial equations $0 = 0$ from $\Phi$. Assume that we derive a new clause $D = \lnot \bigwedge_{i\in I} (g_i = \beta_i)$ by an application of the weakening rule to $C = \lnot \bigwedge_{i\in J} (f_i = \alpha_i)$. We simulate this derivation by syntactic rules as follows:
    \begin{enumerate}
        \item Add equations $g_j = \beta_j$ to $C$ using the syntactic weakening rule, which gives $\lnot \big( \bigwedge_{i\in J} (f_i = \alpha_i) \land \bigwedge_{i\in I} (g_i = \beta_i) \big)$.
        \item Lemma~\ref{fact:well_known} implies that the equation $f_i = \alpha_i$ is a linear combination of the equations $g_j = \beta_j$. Thus, we can subsequently apply the addition rule and get $0 = 0$ from $f_i = \alpha_i$.
        \item We remove the equations $0=0$ from $\lnot \big( \bigwedge_{i\in J} (0 = 0) \land \bigwedge_{i\in I} (g_i = \beta_i)\big)$.
    \end{enumerate}

    After every operation, we obtain a new clause and download it to the memory. Note that to obtain a new clause, we only need the last derived clause and do not need clauses that were derived earlier. Hence, we can delete the previous clause from the memory on every step after we obtain the new one. We do not delete $C$, thus, we need at most one additional clause.
\end{proof}

\begin{lemma}\label{lem:l4}
    Let $\phi$ be an unsatisfiable linear $r$-CNF formula. If there is a $(k + r - 1)$-winning strategy, then the space of every refutation that uses only syntactic derivation rules is at least $k$.
\end{lemma}
\begin{proof}
    Let $\calH$ be a $(k + r - 1)$-winning strategy. We show that if $S_1, \ldots, S_m$ is a $\ResXor$ refutation of $\phi$ (i.e. sequence of configurations) of space less than $k$, then for all $i \in [m]$, every clause from $S_i$ is satisfiable. It is a contradiction since $S_m$ contains the empty clause. We construct by induction on $i$ a sequence of linear systems $\calF_i \in \calH$ such that $\calF_i \vDash S_i$ and $\card{\calF_i} = \card{S_i}$. Moreover, we maintain the following invariant: $\calF_i$ contains at least one equation from every clause from $S_i$.

    The base case is trivial: $\calF_1$ is the empty system.

    For the inductive step suppose that $\calF_{i-1}$ is already defined. We consider four possible cases for $S_i$.

    \begin{enumerate}
        \item
        $S_i$ is obtained from $S_{i-1}$ by a download operation. Thus, $S_i = S_{i-1} \cup \{C\}$, $C \in \phi$, and $C = (f_1 = \alpha_1 \lor \ldots \lor f_n = \alpha_n)$. By the fourth property of $\calH$, we can repeatedly extend $\calF_{i-1}$ with $f_j = a_j$, where $j \in [n]$, and obtain $\calF' \in \calH$. Note that $\card{\calF'} \leq k + r - 1$. By the second property, there exists a solution $\sigma$ of $\calF'$ that satisfies $C$, therefore there exists $j_0 \in [n]$ such that $\alpha_{j_0} = a_{j_0}$, since otherwise every solution of $\calF$ would falsify $C$. Now we delete the equations $f_j = a_j$ for $j \neq j_0$ and obtain $\calF_i$ such that $\card{\calF_i} = \card{\calF_{i-1}} + 1$, $\calF_i \vDash C$, and the invariant holds.

        \item
        $S_i = S_{i-1} \cup \{(C \lor D)\}$ by an inference operation that derives $C \lor D$ from $C \lor (f=0)$ and $D \lor (f=1)$. By the induction hypothesis, $\calF_{i-1}$ contains an equation from $C \lor (f=0)$ and an equation from $D \lor (f=1)$. Since $\calF_{i-1}$ is consistent, it does not contain both $(f = 0)$ and $(f = 1)$ simultaneously. Thus, $\calF_{i-1}$ contains an equation from $C \lor D$. We obtain $\calF_i$ by adding this equation to $\calF_{i-1}$. Then $\card{\calF_i} = \card{\calF_{i-1}} + 1$ and the invariant holds.

        \item
        $S_i = S_{i-1} \cup \{D\}$ by an inference operation, where $C \vDash D$. By assumptions, only syntactic rules can be applied. Thus, we have to consider three subcases:
        \begin{enumerate}
            \item\label{l4_sc2} The syntactic weakening rule that allows deriving $C \lor (f = \alpha)$ from $C$.

            The invariant implies that $\calF_{i-1}$ contains an equation $(g = a)$ from $C$ and we can define $\calF_{i} \coloneqq \calF_{i-1} \land (g = a)$.

            \item\label{l4_sc1} The simplification rule that allows deriving $D$ from $D \lor (0 = 1)$.

            $\calF_{i-1}$ contains an equation from $D\lor(0 = 1)$. This equation cannot be $0 = 1$ since $\calF$ is consistent, so we can copy it as in the previous case.

            \item\label{l4_sc3} The addition rule that allows deriving $D = C' \lor (f_1 = a) \lor (f_1 + f_2 = a_1 + a_2 + 1)$ from $C = C' \lor (f_1 = a_1) \lor (f_2 = a_2)$.

            Similarly, $\calF_{i-1}$ contains an equation from $C' \lor (f_1 = a_1) \lor (f_2 = a_2)$. If this equation is not $f_2 = a_2$, then we can copy it. Otherwise, we extend $\calF_{i-1}$ with $(f_1 = \alpha_1)$. We denote the new system as $\calG$. If $\alpha_1 = a_1$ then we define $\calF_i = \calG$. Otherwise, $\calG = \calF' \land (f_2 = a_2) \land (f_1 = a_1 + 1)$ for some linear system $\calF'$. $\calG$ is equivalent to $\calG' = \calF' \land (f_2 = a_2) \land (f_1 + f_2 = a_1 + a_2 + 1)$. By definition of a winning strategy, $\calG' \in \calH$. We define $\calF_i \coloneqq \calG'$.
        \end{enumerate}

        In all the cases above, $\card{\calF_i} = \card{\calF_{i-1}} + 1$ and the invariant holds.

        \item
        $S_i = S_{i-1} \setminus \{C\}$ by an erasure operation. The invariant implies that $\calF_{i-1} = \calF \land (f=a)$ for some $f=a$ from $C$ and some linear system $\calF$. We define $\calF_i \coloneqq \calF$. $\card{\calF_i} = \card{S_i}$ and the invariant holds since $\calF_i$ contains equations from the other clauses from $S_{i-1}$.

    \end{enumerate}
\end{proof}

\begin{lemma}\label{lem:cl1}
    Let $\phi$ be an unsatisfiable linear $r$-CNF formula. If there is a $(k + r - 1)$-winning strategy, then $\Space(\phi) \geq k-1$.
\end{lemma}
\begin{proof}
    Follows from Lemmas~\ref{lem:prop2} and~\ref{lem:l4}.
\end{proof}

Finally, we are ready to state the connection between $\Space(\phi)$ and $\Width(\phi)$.

\begin{theorem}\label{space_vs_width}
    Let $\phi$ be an unsatisfiable linear $r$-CNF formula. Then $\Space(\phi) \geq \Width(\phi) - r - 1$.
\end{theorem}
\begin{proof}
    Let $\Width(\phi) = k$. If $k \le r+1$, then the statement is trivial. Otherwise, Lemma~\ref{lem:sw_l2} implies that there is a $(k-1)$-winning strategy. Now, by Lemma~\ref{lem:cl1}, $\Space(\phi) \geq k-r-1$.
\end{proof}

\subsection{Applications}

Garl\'{\i}k and Ko\l{}odziejczyk~\cite{Garlik2018} observed that every $\ResXor$ refutation of a CNF formula $\phi$ can be translated into a Polynomial Calculus refutation of its arithmetization. Moreover, this translation can be performed such that the degree of the resulting $\PC_{\bbF_2}$ refutation is equal to the width of the original $\ResXor$ refutation.

\begin{theorem}[follows from the proof of Theorem~18 in \cite{Garlik2018}; also see footnote~4 there]
    Let $\phi$ be an unsatisfiable CNF formula. Then $\Deg(\phi) \le \Width(\phi)$.
\end{theorem}

Thus, every $r$-CNF formula with large degree in $\PC_{\bbF_2}$ also has large $\ResXor$ width. This allows to apply the above results to formulas which have non-trivial lower bounds on $\PC_{\bbF_2}$ degree.

\begin{corollary}\label{space_vs_deg}
    Let $\phi$ be an unsatisfiable $r$-CNF formula. Then $\Space(\phi) \ge \Deg(\phi) - r - 1$.
\end{corollary}

We illustrate this idea by providing space lower bounds for ordering principle and pigeonhole principle.

\subsubsection{Ordering principle}\label{ssec:ord+exp}

In this section we show how Theorem~\ref{space_vs_width} can be used to prove space lower bounds on the ordering principle differently than the one that was shown in Section~\ref{ssec:ord}.
We consider the graph ordering principle which is a generalization of the ordering principle. Let $G = (V, E)$ be a graph. Given a vertex $u$, $N(u)$ is the set of neighbors of $u$. Similarly to the ordering principle, the graph ordering principle says that every ordering of the vertices of $G$ has the minimum. We encode it by a CNF formula $\GOP(G)$ in the variables ${(x_{uv})}_{u \neq v \in V(G)}$, where $x_{uv}$ states that $u$ precedes $v$. 
The clauses of $\GOP(G)$ are the following:
\begin{enumerate}
    \item Anti-symmetry: $\neg x_{uv} \lor \neg x_{vu}$, for all $u \neq v \in V(G)$.
    \item Totality: $x_{uv} \lor x_{vu}$, for all $u \neq v \in V(G)$.
    \item Transitivity: $\neg x_{uv} \lor \neg x_{vw} \lor x_{uw}$, for all distinct $u,v,w \in V(G)$.
    \item Non-minimality (i.e. non-existence of the minimum element):
    \[
        \NM_u \coloneqq \bigvee_{v \in N(u)} x_{vu}, \text{ for all } u \in V(G).
    \]
\end{enumerate}

If $G$ does not have isolated vertices, then $\GOP(G)$ is unsatisfiable.

Now we recall a family of graphs that have non-trivial lower bounds on PC degree.
A graph $G = (V, E)$ is an \emph{$(s, \delta)$-vertex expander} if for every subset of vertices $V' \subset V(G)$ $\card{V'} \leq s$, it holds that $\card{\calN(V')} \geq \delta \card{V'}$, where the neighborhood $\calN(V') = \{v \in V(G)\setminus V': |N(v) \cap V'| > 0\}$ is the set of all vertices in $V(G) \setminus V'$ that have a neighbor in $V'$.


\begin{theorem}[\cite{Galesi2010}, Theorem~1]\label{expand}
    For a non-bipartite graph $G$ that is an $(s, \delta)$-vertex expander it holds thatit holds that $\Deg(\GOP(G)) > \delta s / 4$.
\end{theorem}

\begin{corollary}\label{GOP}
    For a non-bipartite $d$-regular graph $G$ with $d \geq 3$ that is an $(s,\delta)$-vertex expander it holds that $\Space(\GOP(G)) \geq \delta s / 4 - d$.
\end{corollary}
\begin{proof}
    $\GOP(G)$ is a $d$-CNF since the degree of every vertex is $d$. The statement follows from Corollary~\ref{space_vs_deg} and Theorem~\ref{expand}.
\end{proof}

This allows us to obtain a linear lower bound on the space complexity of $\Ordering_n$. Note that this bound is weaker than the one from Theorem~\ref{thm:ordering_bound}.

\begin{theorem}\label{GOPGOP}
    $\Space(\Ordering_n) = \Omega(n)$.

\end{theorem}
\begin{proof}
    It is known that there exist $d \geq 3$, $\delta > 0$, $\epsilon > 0$, and $N$ such that for every $n \geq N$ there exists a $d$-regular $(\epsilon n, \delta)$-vertex expander $G$ on $n$ vertices (e.g. see~\cite{hoory06}, Theorem~4.16).

    Observe that $\GOP(G)$ can be obtained by setting some variables in $\Ordering_{n}$ to $0$.
    
    Thus, $\Space(\Ordering_{n}) \geq \Space(\GOP(G)) \geq \delta s / 4 - d$, where the last inequality follows from Corollary~\ref{GOP}.
\end{proof}


\subsubsection{Functional pigeonhole principle}

Functional pigeonhole principle is a generalization of the pigeonhole principle that also assumes functionality axioms, which disallow placing the same pigeon in several holes. Note that the proof of Lemma~\ref{lem:small_system_php} seems to fail with these axioms present. Moreover, since we need a formula with small initial width, we consider the graph variant of this principle.

Let $G = (U \sqcup V, E)$ be a bipartite graph. We define the CNF formula $\FPHP(G)$ consisting of the following clauses:
\begin{enumerate}
    \item Pigeon axioms: $P_u \coloneqq \bigvee_{v \in N(u)} p_{uv}$, for all $u \in U$.
    \item Hole axioms: $H_{uu'v} \coloneqq \lnot p_{uv} \lor \lnot p_{u'v}$, for all $u \neq u' \in U$ and all $v \in V$.
    \item Functionality axioms: $F_{uvv'} \coloneqq \lnot p_{uv} \lor \lnot p_{uv'}$, for all $u \in U$ and all $v \neq v' \in V$.
\end{enumerate}

In addition, we denote the classical functional pigeonhole principle as $\FPHP_n^{n+1}$, which is equivalent to $\FPHP(K_{n+1,n})$.

Now recall that a bipartite graph $G = (U \sqcup V, E)$ is a \emph{bipartite $(s, \delta)$-boundary expander} if for each vertex set $U' \in U$, $\card{U'} \leq s$, it holds that $\card{\partial(U')} \geq \delta \card{U'}$, where $\partial(U')$ is the set of all vertices in $V$ that have a unique neighbor in $U'$.

\begin{theorem}[\cite{Nordstrom2015}, Theorem~4.9]
    Suppose that $G = (U \sqcup V, E)$ is a bipartite $(s, \delta)$-boundary expander with left degree bounded by $d$. Then it holds that refuting $\FPHP(G)$ in polynomial calculus requires degree strictly greater than $\delta s / (2d)$.
\end{theorem}

\begin{theorem}
    $\Space(\FPHP_n^{n+1}) = \Omega(n)$.
\end{theorem}
\begin{proof}
    It is known that there exist $\epsilon > 0$ and $\delta > 0$ such that there exists a family of bipartite $(\epsilon n, \delta)$-boundary expanders with left-degree $3$ (e.g. see~\cite{Nordstrom2015} and the references given there). It is easy to see that for every bipartite $G = (U \sqcup V, E)$ with $\card{U} = n+1$ and $\card{V} = n$, $\FPHP(G)$ can be obtained from $\FPHP_{n}^{n+1}$ by setting some variables to $0$. 

    The rest of proof is analogous to the proof of Corollary~\ref{GOP}.
\end{proof}

\begin{acks}
    The research is supported by \grantsponsor{RCSF}{Russian Science Foundation}{https://rscf.ru/en/} (project~\grantnum{RCSF}{18-71-10042}).
\end{acks}

\bibliographystyle{ACM-Reference-Format}
\bibliography{refs}


\end{document}